\documentclass[submission,copyright,creativecommons]{eptcs}
\providecommand{\event}{AFL 2023} 

\usepackage{iftex,xcolor}
\ifpdf
  \usepackage{underscore}         
  \usepackage[T1]{fontenc}        
\else
  \usepackage{breakurl}           
\fi

\title{Pumping Lemmata for Recognizable Weighted Languages \\
  over \textsc{Artinian} Semirings}
\author{Andreas Maletti
  \institute{Universit\"at Leipzig \\
    Faculty of Mathematics and Computer Science \\
    PO~Box~100\,920, 04009~Leipzig, Germany}
  \email{maletti@informatik.uni-leipzig.de}
  \and
  Nils Oskar Nuernbergk
  \email{nils.nuernbergk@gmail.com}
}

\def\titlerunning{Pumping Lemmata for \upshape
    \textsc{Artinian} Semirings}
\def\authorrunning{A.~Maletti, N.~O.~Nuernbergk}

\usepackage{enumitem}

\setenumerate{itemsep=0pt, topsep=0pt}
\setenumerate[1]{label=(\roman*)}

\usepackage[ngerman,english]{babel}
\usepackage{amsmath,amssymb,amsthm,mathtools}

\useshorthands{"}
\addto\extrasenglish{\languageshorthands{ngerman}}
\allowdisplaybreaks

\usepackage{stmaryrd} 

\theoremstyle{plain}
\newtheorem{thm}{Theorem}[section]
\newtheorem{lem}[thm]{Lemma}
\newtheorem*{cor}{Corollary}

\theoremstyle{definition}
\newtheorem{defnx}[thm]{Definition}
\newtheorem{exmpx}[thm]{Example}

\newenvironment{defn} 
{%
\pushQED{\qed}\begin{defnx}}
{\popQED\end{defnx} %
}

\newenvironment{exmp} 
{%
\pushQED{\qed}\begin{exmpx}}
{\popQED\end{exmpx}
}

\let\leq=\leqslant
\let\geq=\geqslant

\DeclarePairedDelimiter{\abs}{\lvert}{\rvert}
\DeclarePairedDelimiter{\braces}{\lbrace}{\rbrace}
\DeclarePairedDelimiter{\angles}{\langle}{\rangle}
\DeclareMathOperator{\Ker}{Ker}

\newcommand{\NN}{\mathbb{N}_{\scriptscriptstyle +}}
\newcommand{\N}{\mathbb N}
\newcommand{\ZZ}{\mathbb{Z}}
\newcommand{\QQ}{\mathbb{Q}}
\newcommand{\BB}{\mathbb{B}}
\providecommand{\seq}[3]{\ensuremath{#1_{#2}, \dotsc, #1_{#3}}}

\newcommand{\inn}{\operatorname{in}}
\newcommand{\out}{\operatorname{out}}

\newcommand{\supp}{\operatorname{supp}}
\newcommand{\im}{\operatorname{im}}
\newcommand{\Qmax}{\QQ^{\max}}
\newcommand{\End}{\operatorname{End}}
\newcommand{\Hom}{\operatorname{Hom}}
\newcommand{\id}{\operatorname{id}}
\newcommand{\0}{\braces{0}}

\newcommand\restrict[1]{\raisebox{-.3ex}{$|$}_{#1}}

\begin{document}
\maketitle

\begin{abstract}
  Pumping lemmata are the main tool to prove that a certain language
  does not belong to a class of languages like the recognizable
  languages or the context-free languages.  Essentially two pumping
  lemmata exist for the recognizable weighted languages: the classical
  one for the \textsc{Boolean} semiring (i.e., the unweighted case),
  which can be generalized to zero-sum free semirings, and the one for
  fields.  A joint generalization of these two pumping lemmata is
  provided that applies to all \textsc{Artinian} semirings, over which
  all finitely generated semimodules have a finite bound on the length
  of chains of strictly increasing subsemimodules.  Since
  \textsc{Artinian} rings are exactly those that satisfy the
  Descending Chain Condition, the \textsc{Artinian} semirings include
  all fields and naturally also all finite semirings (like the
  \textsc{Boolean} semiring).  The new pumping lemma thus covers most
  previously known pumping lemmata for recognizable weighted
  languages.
\end{abstract}

\section{Introduction}
The class of recognizable languages~\cite{yu97} is certainly the
best-studied and one of the most useful classes of languages.  It has
excellent closure properties, and all standard decision problems for it
are decidable.  Applications of the recognizable languages are too
numerous to list, but include pattern
matching~\cite[Chapter~10]{focs}, lexical analysis~\cite{ahoull85},
input validation~\cite{loojon21}, network protocols~\cite{icc}, and
DNA sequence analysis~\cite{bio}.  Pumping lemmata are statements of the
form that given a suitably long word in the language, we can always
identify a subword that can be iterated (or pumped) at will without
leaving the language.  Such statements exist for many language
classes including the recognizable~\cite{yu97} and context-free
languages~\cite{autberboa97}, and they allow a relatively
straightforward proof that a given language does not belong to the
class (e.g., is not recognizable).

In several applications~\cite{albjar09,baigrocie09,knimay09}, the
purely qualitative yes/no"~decision of languages is completely
insufficient.  This led to the introduction of weighted
languages~\cite{sch61} (see~\cite{sak09} for an excellent survey), in
which each word is assigned a weight from a
semiring~\cite{hebwei98,gol99}.  The classical recognizable languages
are reobtained by considering the support of the recognizable weighted
languages over the \textsc{Boolean} semiring~$(\{0,1\},
\mathord{\max}, \mathord{\min}, 0, 1)$.  The theory of recognizable
weighted languages is also very well developed and several
textbooks~\cite{salsoi78,kuisal86,droste2009handbook} provide
excellent introductions.

Determining whether a given weighted language is recognizable is often
even more difficult than in the unweighted case, and we again mostly
rely on pumping lemmata~\cite{JAC80,REUT80} to prove that a given
weighted language is not recognizable.  However, the coverage
situation is very unsatisfactory.  The classical pumping lemma for
unweighted languages can be lifted to all zero-sum free
semirings~\cite{hebwei98,gol99} (i.e., semirings in which~$a + b = 0$
implies~$a = 0 = b$) by means of a semiring homomorphism from such a
semiring into the \textsc{Boolean} semiring~\cite{wan97} and a
construction that avoids zero-divisors~\cite{kir11}.  On the other
hand, the pumping lemmata of~\cite{JAC80,REUT80} require the semiring
to be a field, which necessarily is not zero-sum free.  Despite their
similarities, the two recalled pumping lemmata thus apply to completely
disjoint sets of semirings, which do not even cover all semirings
(e.g., the finite ring~$\mathbb Z_4$ is not zero-sum free and not a
field).  Indeed it is well-known~\cite{drokus21} how to handle finite
semirings like~$\mathbb Z_4$ (by encoding the weights into the
states), so that the classical unweighted pumping lemma becomes
applicable.  Similarly, it is known how to handle semirings
like~$\mathbb Z$ that embed into a field, but there are also infinite
semirings that are not zero-sum free and not (embeddable into) a field
like the ring~$\QQ[x]/(x^2)$ of rational linear polynomials.  The
ring~$\QQ[x]/(x^2)$ cannot embed into a field since it has zero-divisors
(e.g.,~$x \cdot x = 0$), but it fulfills the requirements for our
pumping lemma.  Hence there are semirings for which we currently have
no available pumping lemma, as well as different semirings that permit
essentially the same pumping lemma for their recognizable weighted
languages but with totally different justifications.

Let us recall the statement of these pumping lemmata.  Let~$L \colon
\Sigma^* \to S$ be a recognizable weighted language, which assigns to
each word~$w \in \Sigma^*$ a weight~$L(w) \in S$ in the semiring~$S$.
The support of~$L$ is the set~$\supp L = \{w \in \Sigma^* \mid L(w)
\neq 0\}$ of nonzero-weighted words in $L$.  The pumping lemma states
that given a sufficiently long word~$w \in \supp L$, there exists a
decomposition~$w = uxv$ such that~$ux^kv \in \supp L$ for
infinitely many~$k \in \N$.  In other words, $ux^kv$~is also
nonzero-weighted in~$L$ for infinitely many~$k \in \N$, where~$ux^kv =
ux\dotsm xv$ with $k$~repetitions of~$x$.

In this contribution we will establish such a pumping lemma for a
class of semirings that includes all fields and all finite semirings.
Thus, we directly cover both the pumping lemmata
of~\cite{JAC80,REUT80} as well as the classical pumping
lemma~\cite[Lemma~2]{PUMP}.  We achieve this by following the general
approach of~\cite{REUT80} while trying to avoid the vector space
structure utilized there.  This requires some minor adjustments and,
in particular, a replacement for the dimension, for which we use the
length of a semimodule.  A semimodule has finite length if there is a
finite bound on the length of strictly increasing chains of
subsemimodules.  This notion also allows us to define the
\textsc{Artinian} semirings that we consider.  A semiring is
\textsc{Artinian} if each finitely generated semimodule has finite
length.  The \textsc{Artinian} semirings include all fields and all
finite semirings, but not all zero-sum free semirings.  However, the
mentioned approach for zero-sum free semirings (applying the
homomorphism into the \textsc{Boolean} semiring and avoiding
zero-divisors) naturally also works with our pumping lemma.

We first show that any endomorphism of a semimodule over an
\textsc{Artinian} semiring is surjective if and only if it is
injective, which is a generalization of a well-known statement for
vector spaces.  Following the approach of~\cite{REUT80}, we introduce
pseudoregular endomorphisms using 2~of the 5~characterizing properties
utilized in~\cite[Proposition~1]{REUT80}.  Fortunately, these are the
two main properties needed for the proof of our pumping lemma, and the
remaining 3~properties rely on infrastructure that is not generally
available in our semimodules (instead of the vector spaces used
in~\cite{REUT80}).  The argument that a sufficiently long composition
of endomorphisms needs to contain a pseudoregular endomorphism can be
taken over mostly unchanged from~\cite{JAC80}, which then almost
directly yields our main pumping lemma.  Finally, we also briefly
consider pumping lemmata for infinite alphabets.

\section{Preliminaries}
\label{sec:prelim}
We denote the non-negative integers by~$\N$ and the positive integers
by~$\NN = \N \setminus \0$.  Moreover, we let~$\QQ^{\geq 0} =
\braces{q \in \QQ \mid q \geq 0}$ be the set of non-negative rational
numbers.  For every alphabet~$\Sigma$ we denote the free monoid
over~$\Sigma$ by~$\Sigma^*$, i.e., $\Sigma^*$~is the set of all finite
words with letters in~$\Sigma$.  We write~$\varepsilon$ for the empty
word (the neutral element of the free monoid).  Additionally, we
let~$\Sigma^{\scriptscriptstyle+} = \Sigma^* \setminus
\braces{\varepsilon}$.  For all sets~$A$, $B$, and~$C$ and
all maps $f \colon A \to B$~and~$g \colon B \to C$, we
let~$\mathord{\id_A} = \braces{(a, a) \mid a \in A}$ and~$(gf) \colon
A \to C$ be the map such that~$(gf)(a) = g(f(a))$ for every~$a \in
A$.  Finally, if~$A = B$, then we let~$f^0 = \mathord{\id_A}$
and~$f^{k+1} = ff^k$ for every~$k \in \N$.

A \emph{(commutative) semiring}~\cite{hebwei98,gol99} is an algebraic
structure~$(S, \mathord+, \mathord\cdot, 0, 1)$, in which $S$~is a
set, called \emph{carrier}, $(S, \mathord+, 0)$~and~$(S,
\mathord\cdot, 1)$ are commutative monoids, called \emph{additive} and
\emph{multiplicative monoid} respectively, and 
\begin{align*}
  r \cdot (s + t)
  &= (r \cdot s) + (r \cdot t), \tag{distributivity} \\*
  0 \cdot r
  &= 0 \tag{absorption of~$0$}
\end{align*}
for all~$r, s, t \in S$.  We will refer to the semiring~$(S,
\mathord+, \mathord\cdot, 0, 1)$ simply by its carrier set~$S$ and
denote multiplication by juxtaposition as usual.  For
the rest of the contribution, let $S$~be a commutative semiring.

A (commutative) \emph{ring} is simply a semiring in which every
element has an additive inverse, and a (commutative) \emph{semifield}
is similarly a semiring in which every element~$s \in S \setminus
\{0\}$ has a multiplicative inverse.  As usual, a (commutative)
\emph{field} is a ring that is also a semifield.  The \textsc{Boolean}
semifield is~$\BB = (\braces{0, 1}, \mathord{\max},
\mathord{\min}, 0, 1)$.

An $S$-semimodule~\cite{hebwei98,gol99} is a tuple~$(M, \mathord\oplus,
0_M, \mathord\odot)$ consisting of a commutative monoid~$(M,
\mathord\oplus, 0_M)$ and a mapping~$\mathord\odot \colon S \times
M \to M$ such that
\begin{align*}
  (r \cdot s) \odot u
  &= r \odot (s \odot u), \tag{associativity} \\*
  r \odot (u \oplus v)
  &= (r \odot u) \oplus (r \odot v), \tag{left
    distributivity} \\*
  (r + s) \odot u
  &= (r \odot u) \oplus (s \odot u), \tag{right
    distributivity} \\*
  0 \odot u
  &= 0_M \tag{absorption of~$0$}
\end{align*}
for all semiring elements~$r, s \in S$, also called \emph{scalars},
and semimodule elements~$u, v \in M$.  As before, we write just~$M$
for the semimodule~$(M, \mathord\oplus, 0_M, \mathord\odot)$,
and due to the compatibility axioms presented above, we can safely
stop distinguishing the semimodule addition~$\oplus$ and semiring
addition~$+$, writing just~$+$ for both, as well as mixed
multiplication~$\odot$ and semiring multiplication~$\cdot$,
writing~$\cdot$ for both, and the additive neutral element~$0_M$ of
the semimodule and its corresponding element~$0$ of the semiring,
writing~$0$ for both.  Finally, we let~$su = s \cdot u$ for
all~$s \in S$ and~$u \in M$.  It is clear that the semiring~$S$ itself
forms a semimodule, semimodules over rings
are simply modules, and semimodules over fields are vector spaces.  A
\emph{subsemimodule} of~$M$ is a subset~$N \subseteq M$ such that~$0
\in N$, $m + n \in N$ for all~$m, n \in N$, and $r \cdot n \in N$ for
all~$r \in S$ and~$n \in N$.  In other words, a subsemimodule is a
subset that forms a semimodule itself with respect to the operations
of~$M$ suitably restricted to~$N$.  We write~$N \preceq M$ if $N$~is a
subsemimodule of~$M$.  For every subset~$V \subseteq M$ we
write~$\angles{V}$ for the \emph{span} of~$V$ (i.e., the smallest
subsemimodule of~$M$ that contains~$V$) and say that~$\angles{V}$ is
\emph{generated} by~$V$.

Let $M$~and~$N$ be two semimodules and~$\varphi \colon M \to N$ a
mapping. Then $\varphi$~is \emph{linear} (or a \emph{semimodule
  homomorphism}) if
\[ s \cdot \varphi(u) = \varphi(s \cdot u) \qquad \text{and} \qquad
  \varphi(u + v) = \varphi(u) + \varphi(v) \] for all~$s \in S$
and~$u, v \in M$.  Note that~$\varphi(0) = 0$ if $\varphi$ is linear
by the former condition.  If $\varphi$~is bijective and linear, then
we call~$\varphi$ an \emph{isomorphism} and say that $M$~and~$N$ are
\emph{isomorphic}, which we write as~$M \cong N$.  We
let~$\ker \varphi = \braces{ m \in M \mid \varphi(m) = 0}$ be the
\emph{kernel} of~$\varphi$
and~$\im \varphi = \braces{\varphi(m) \mid m \in M }$ be the
\emph{image} of~$\varphi$ in~$N$, which is always a subsemimodule
of~$N$ provided that~$\varphi$ is linear.  The first isomorphism
theorem~\cite[p. 162, Corollary 5.16]{ALUFFI} states that $M/\ker
\varphi \cong \im \varphi$ for every ring~$S$ and linear
map~$\varphi$.  Here, $M/\ker \varphi$~is the set of equivalence
classes~$M/\sim$ with the equivalence relation~$\sim$ given
by~$m \sim n$ if~$m - n \in \ker \varphi$ and addition and scalar
multiplication defined by~$[m] + [n] = [m + n]$ and $s[m] = [sm]$
(where $[m]$~denotes the equivalence class of~$m$).  Thus, over a
ring~$S$, the linear map~$\varphi$ is injective if and only if~$\ker
\varphi = \0$.  Moreover, we let 
\[ \Hom(M, N) = \braces{\varphi \colon M \to N \mid \varphi \text{ is
      linear}}, \quad \End(M) = \Hom(M, M), \quad
  \text{and} \quad M^\vee = \Hom(M, S), \] which form
semimodules with pointwise addition and scalar multiplication.  The
semimodule~$\End(M)$ contains the \emph{endomorphisms} of~$M$, and
$M^\vee$ is called the \emph{dual semimodule} of~$M$.

Let $Q$~be an arbitrary set.  Then
\[ S^Q = \braces{ f \colon Q \to S \mid \ker f \text{ is
      co-finite} } \]
forms a semimodule with pointwise addition and scalar multiplication
that we call the \emph{free semimodule over $Q$} (unique up to
isomorphism as usual).  This is justified by the fact~\cite[p. 194]{gol99}
that for any semimodule~$M$ every mapping~$\varphi \colon Q \to M$
uniquely extends to a linear map~$\widetilde{\varphi} \colon S^Q \to
M$ such that~$\widetilde{\varphi}(\iota_q) = \varphi(q)$, where
$\iota_q \in S^Q$~is the mapping given for every~$p \in Q$ by
\[ \iota_q(p) =
  \begin{cases}
    1
    & \text{if } p = q \\*
    0 & \text{otherwise.}
  \end{cases}
\]
In particular, if~$\im \varphi$ generates~$M$, then
$\widetilde{\varphi}$~is surjective.  If $S$ is a field,
then every semimodule (i.e., vector space) is free, but the same
is not true for arbitrary semirings~$S$.  If $Q$~is finite, then we
say that $S^Q$~is of \emph{rank}~$n = \abs{Q}$ and will often
identify~$S^Q$ with the semimodule~$S^n$.  

The spaces~$\Hom(M, N)$, $\End(M)$, and~$M^\vee$ are particularly easy
to describe when $M$~and~$N$ are free of finite rank~\cite[p. 195]{gol99}.
These are exactly the matrix spaces
\[\Hom(S^Q, S^P) \cong S^{P \times Q}, \qquad 
  \End(S^Q) \cong S^{Q \times Q}, \qquad \text{and} \qquad
  (S^Q)^\vee \cong S^{\braces{1}\times Q} \cong S^Q.
\]
Note also that $S^Q$~itself can be identified with the matrix
space~$S^{Q \times \braces{1}}$.  Matrix multiplication (i.e.,
composition of linear maps) is then defined as follows: for every~$M
\in S^{P \times Q}$ and~$N^{Q \times R}$, the matrix~$M \cdot N \in
S^{P \times R}$ is given for all~$p \in P$ and~$r \in R$ by
\[ (M \cdot N)_{pr} = \sum_{q \in Q} M_{pq} \cdot N_{qr}. \]
We will usually state theorems in terms of linear maps instead of
matrices due to their greater generality (non-free semimodules do not
generally permit descriptions by matrices) and clarity of
presentation.

Let $\Sigma$~be an alphabet.  A \emph{weighted language} over~$\Sigma$
is a function~$L \colon \Sigma^* \to S$.  Given~$w \in \Sigma^*$ and a
weighted language~$L \colon \Sigma^* \to S$, we occasionally
write~$L_w$ instead of~$L(w)$.  The \emph{support} of~$L$ is the
set~$\supp L = \braces{w \in \Sigma^*
  \mid L_w \neq 0}$.

A \emph{linear representation}~\cite{drokus21} of a weighted
language~$L \colon \Sigma^* \to S$ is a tuple~$(Q, \mathord{\inn},
\mathord{\out}, \mu)$, where $Q$~is a finite set of \emph{states}, 
$\mathord{\inn} \in (S^Q)^\vee$ is an \emph{input vector},
$\mathord{\out} \in S^Q$ is an \emph{output vector}, and~$\mu \colon
\Sigma^* \to \End(S^Q)$ is a monoid homomorphism (where the monoid
structure on~$\End(S^Q)$ is given by composition of maps), such that
for every~$w \in \Sigma^*$
\[ L_w = \mathord{\inn} \cdot \mu(w) \cdot \mathord{\out}. \]
If a weighted language~$L$ admits a linear representation, then we
call~$L$ \emph{recognizable}.  This definition of recognizability is
equivalent to other common definitions given in terms of weighted
automata~\cite{sak09}.

\section{Semimodules of Finite Length}
\label{sec:semimodules}
We recall that the dimension of a finite dimensional vector space~$V$
provides an upper bound on the number of proper inclusions in any
chain of subspaces of~$V$; i.e., if $V_0 \preceq \dotsb \preceq V_r$
is a chain of subspaces of~$V$ and~$r > \dim V$, then there is at
least one~$0 \leq i < r$ such that~$V_i = V_{i+1}$.

In this spirit, we define the \emph{length}~$\ell(M) \in \N \cup
\{\infty\}$ of a semimodule~$M$ to be the (possibly infinite) least
upper bound on the number of proper inclusions in any chain of
subsemimodules of~$M$; i.e., 
\[ \ell(M) = \sup \braces{r \mid M_0 \prec \dotsb \prec M_r \text{ is
      a chain of strictly increasing subsemimodules of~$M$}}. \]
Clearly, $\dim V = \ell(V)$~for every finite dimensional vector
space~$V$.  However, the length is distinct from the rank of a free
module even if $S$~is a ring.  For example, $\ZZ$~has rank~$1$ as a $\ZZ$-module, but~$\ell(\ZZ) = \infty$ since
\[ \angles{k^m} \prec \angles{k^{m-1}} \prec \dotsb \prec \angles{k} \] 
is a chain of strictly increasing submodules of~$\ZZ$ for every~$k \in
\ZZ \setminus \braces{0, 1, -1}$ and~$m \geq 2$.

\begin{defn}
  We say that a semimodule~$M$ has \emph{finite length} if $\ell(M)
  \in \N$; i.e., $\ell(M)$ is finite.
\end{defn}

Let us provide some examples of semimodules that have finite length.

\begin{exmp}
  \mbox{ }
  \begin{enumerate}
  \item Finite dimensional vector spaces over fields have finite length.
  \item Finite semimodules have finite length.
  \item We consider the commutative monoid~$M = \QQ^{\geq 0} \cup
    \braces{\infty}$ with~$u + \infty = \infty$ for all~$u \in M$ and
    addition defined as in~$\QQ$ otherwise.  Then~$M$ is a semimodule
    over~$\QQ^{\geq 0}$ via
    \[ m \odot u = \begin{cases} 
        0
        & \text{if $m = 0$} \\*
        \infty
        & \text{if $m \neq 0$ and $u = \infty$} \\*
        m \cdot u
        & \text{otherwise.}
      \end{cases}
    \]
    We can easily see that the only subsemimodules of~$M$ are~$\0$,
    $\braces{0, \infty}$, $\QQ^{\geq 0}$ and~$M$ itself.  By considering the inclusions among these subsemimodules, we
    obtain~$\ell(M) = 2$.  Notably, this is an example of an infinite
    semimodule that has finite length, but cannot be embedded into a
    module over a ring.  The embedding
    fails since $\infty$~is additively absorptive (i.e., $u + \infty =
    \infty$ for all~$u \in M$, which yields that~$\infty$ cannot be
    inverted).  \qedhere
  \end{enumerate}
\end{exmp}

Let $M$~be a semimodule that has finite length.  Next we show that the
image~$\im \varphi$ of a linear map~$\varphi \colon M \to N$
necessarily has finite length as well. 

\begin{lem}
  \label{finitemap}
  Let $M$~and~$N$ be semimodules and~$\varphi \colon M \to N$ be a
  linear map.  Then~$\ell(\im \varphi) \leq \ell(M)$.
\end{lem}

\begin{proof}
  If~$\ell(M) = \infty$, then the statement holds automatically.
  Therefore, suppose that~$\ell(M) \in \N$ is finite.  We recall that the
  preimage~$\varphi^{-1}(L)$ of a subsemimodule~$L \preceq N$ is a
  subsemimodule of~$M$.  To see this, let~$u, v \in \varphi^{-1}(L)$.
  Then~$\varphi(u + v) = \varphi(u) + \varphi(v) \in L$, and
  thus~$u + v \in \varphi^{-1}(L)$.  Similarly, for every~$s \in S$ we
  have~$\varphi(su) = s\varphi(u) \in L$, and
  thus~$su \in \varphi^{-1}(L)$. Thus, any
  chain~$N_0 \preceq \dotsb \preceq N_r$ of subsemimodules of~$\im
  \varphi$ induces a chain~$\varphi^{-1}(N_0) \preceq \dotsb \preceq 
  \varphi^{-1}(N_r)$ of subsemimodules of~$M$.  Next, we establish
  that~$\varphi^{-1}(N_i) \prec \varphi^{-1}(N_{i+1})$ for every~$0
  \leq i < r$ such that~$N_i \prec N_{i+1}$.  To this end, let~$u \in
  N_{i+1} \setminus N_i$ and select~$v \in \varphi^{-1}(\braces{u})$,
  which is possible because~$N_{i+1} \preceq \im \varphi$.
  Obviously,~$v \notin \varphi^{-1}(N_i)$, which proves that~$v \in
  \varphi^{-1}(N_{i+1}) \setminus \varphi^{-1}(N_i)$ and
  thus~$\varphi^{-1}(N_i) \prec \varphi^{-1}(N_{i+1})$.  Hence,
  $\ell(\im \varphi) \leq \ell(M)$ follows immediately from the
  definition.
\end{proof}

The preceding lemma already suggests that semimodules of finite length
share nice properties with finite dimensional vector spaces.  In order
to harness these, it would be very desirable for the class of finite
length semimodules to have good closure properties.  However, it is
not even closed under direct sums, as the following example
demonstrates.

\begin{exmp}
  Consider the semifield~$S = \Qmax = (\QQ^{\geq 0}, \mathord{\max},
  \mathord{\cdot}, 0, 1)$.  As usual, $\Qmax$~is a semimodule over
  itself, and the presence of multiplicative inverses immediately
  yields that~$\ell(\Qmax) = 1$ because its only subsemimodules are
  $\0$~and~$\Qmax$: if~$H \preceq \Qmax$ and~$H \neq \0$, there is
  an~$h \in H$ with~$h \neq 0$, so~$s = s \cdot h^{-1} \cdot h \in
  \Qmax$ for all~$s \in \Qmax$; whereby~$H = \Qmax$ (indeed, this
  argument works for any semifield). 
  
  Now we consider the direct sum~$M = \Qmax \oplus \Qmax$ of two
  copies of~$\Qmax$, which consists of pairs of rational numbers with
  the maximum applied coordinate-wise.   Clearly, $M$~is also a
  $\Qmax$"~semimodule via a coordinate-wise product.  However,
  $M$~does not have finite length over~$\Qmax$ by the following
  lemma.
\end{exmp}

\begin{lem} \label{counterexample}
  The $\Qmax$"~semimodule~$\Qmax \oplus \Qmax$ has length~$\ell(\Qmax
  \oplus \Qmax) = \infty$.
\end{lem}
\begin{proof}
  Let $M = \Qmax \oplus \Qmax$.  First, we define the function~$q
  \colon M \to \QQ$ such that~$q \bigl(\langle a, b\rangle \bigr) =
  \tfrac ab$ for every~$a, b \in \Qmax$.  Obviously,
  \begin{equation}
    \label{eq:1}
    q \bigl(s\langle a, b\rangle \bigr) = q \bigl(\langle sa,
    sb\rangle \bigr) = \frac{sa}{sb} = \frac ab = q \bigl(\langle a,
    b\rangle \bigr)
  \end{equation}
  for all~$\langle a, b\rangle \in M$ and~$s \in \Qmax$.
  Additionally, for all $\langle a, b\rangle, \langle c,d\rangle \in
  M$ we have
  \begin{equation}
    \label{eq:2}
    q \Bigl(\max \bigl(\langle a, b\rangle, \langle c,d\rangle \bigr)
    \Bigr) \leq \max \Bigl(q \bigl(\langle a, b\rangle \bigr), q
    \bigl(\langle c,d\rangle \bigr) \Bigr)
  \end{equation}
  because
  \[ \frac a{\max(b,d)} \leq \frac ab = q \bigl(\langle a, b\rangle
    \bigr) \qquad \text{and} \qquad \frac c{\max(b,d)} \leq \frac cd =
    q \bigl(\langle c,d\rangle \bigr), \] 
  which yield
  \[ q \Bigl( \max \bigl(\langle a, b\rangle, \langle c,d\rangle
    \bigr) \Bigr) = \frac{\max(a, c)}{\max(b,d)} = \max \Big(
    \frac{a}{\max(b,d)}, \frac{c}{\max(b,d)} \Big) \leq \max \Bigl(q
    \bigl(\langle a, b\rangle \bigr), q \bigl(\langle c,d\rangle
    \bigr) \Bigr). \] For every~$i \in \N$
  let~$u_i = \langle i, 1\rangle$
  and~$M_i = \angles{\braces{\seq u0i}}$ be the subsemimodule
  generated by~$\braces{\seq u0i}$.  Due to the properties
  \eqref{eq:1}~and~\eqref{eq:2} of~$q$, we have~$q(u) \leq q(u_i)$ for
  every~$u \in M_i$.  This immediately yields~$M_i \prec M_{i+1}$ for
  every~$i \in \N$ and thus~$M_0 \prec \dotsb \prec M_i \prec \dotsb$
  is an infinite chain of strictly increasing
  subsemimodules.\footnote{In 
    fact, this is an example of a more general pathology of
    semimodules.  Finite length semimodules are \textsc{Noetherian}
    since they satisfy the Ascending Chain Condition (i.e., every
    ascending chain of subsemimodules terminates).  This proof
    demonstrates that \textsc{Noetherian} semimodules, unlike
    \textsc{Noetherian} modules over rings, are not closed under
    direct sums.  The same is true for the Descending Chain Condition,
    which can be seen by setting $v_i = (1, i)$~and~$N_i =
    \angles{\braces{\seq v0i}}$ for all $i \in \mathbb{N}$.  Then the
    chain~$N_0 \succ \dotsb \succ N_i \succ \dotsb$ does not terminate
    by the same argument as above.}
\end{proof}

Fortunately, for rings~$S$ the situation does not look nearly as
bleak and the expected equalities for length hold, as expressed in the
next theorem. 

\begin{thm}
  \label{artinrings}
  Suppose that $S$~is a ring.
  \begin{enumerate}
  \item Let $M$~and~$N$ be modules such that~$N \preceq M$.  If
    $N$~and~$M/N$ both have finite length, then $M$~has finite length
    and~$\ell(M) = \ell(N) + \ell(M/N)$.
  \item If $M$~and~$N$ are finite-length modules, then $\ell(M \oplus
    N) = \ell(M) + \ell(N)$.
  \item If $S$~has finite length, then every module generated by
    $n \in \N$~elements has length at most~$n \cdot \ell(S)$.
  \end{enumerate}
\end{thm}

\begin{proof}
  \mbox{ }
  \begin{enumerate}
  \item The proof idea for the inequality~$\ell(M) \leq \ell(N) +
    \ell(M/N)$ draws from a proof of the analogous fact for
    \textsc{Noetherian} rings~\cite[10f, Proposition 3.3]{MILNE}.  For
    the sake of a contradiction, assume that $\ell(M) \geq r$,
    where~$r = \ell(N) + \ell(M/N) + 1$.   Then there exists a
    chain~$L_0 \prec \dotsb \prec L_r$ of strictly increasing
    submodules of~$M$.  In the corresponding chain  
    \[ \frac{N + L_0}{N} \preceq \dotsb \preceq \frac{N + L_r}{N} \]
    of $r+1$~submodules of~$M/N$, at most~$\ell(M/N)$ inclusions are
    proper, so $\ell(N) + 1$~inclusions are not.  Similarly, in the
    chain~$L_0 \cap N \preceq \dotsb \preceq L_r \cap N$ of submodules
    of~$N$, at most~$\ell(N)$ inclusions are proper, so $\ell (M/N) +
    1$~inclusions are not.  By the pigeonhole principle, there
    exists~$0 \leq i < r$ such that
    \[ L_i \cap N = L_{i+1} \cap N \quad \text{and} \quad N + L_i = N
      + L_{i+1}. \]
    We note that the latter result relies on the fact that $N \leq H
    \leq K$~and~$H/N = K/N$ together imply~$H = K$~(since if $k \in K$ and $[h] = [k]$ for some $h \in H$, then $k - h \in N \preceq H$, so $k = h + (k - h) \in H$). 

    Now, let~$u \in L_{i+1}$ be arbitrary.  By the second equation
    above, we have~$u = n + v$ for some~$n \in N$ and~$v \in L_i$.  This
    yields~$n = u - v \in L_{i+1} \cap N = L_i \cap N$ and thus~$u \in
    L_{i}$.  Therefore,~$L_{i+1} = L_i$, which is the desired
    contradiction.

    Thus, we have shown that~$\ell(M) \leq \ell(N) + \ell(M/N)$.  It
    remains to show the converse inequality.  Note that any submodule
    of~$M/N$ has the form~$L/N$ for some~$L \preceq M$ such that~$N
    \preceq L$ since $N$~is the preimage of~$[0] \in L/N$.  This claim
    was already shown in a more general setting in the proof of
    Lemma~\ref{finitemap}.  Therefore, let $N_0 \prec \dotsb \prec
    N_n$ with~$n = \ell(N)$~and~$L_{0}/N \prec \dotsb \prec L_m/N$
    with~$m = \ell(M/N)$ be chains of strictly increasing submodules
    of $N$~and~$M/N$, respectively, which exist by the definition of
    the respective length.  These chains can be concatenated to
    obtain a chain
    \[ N_0 \prec \dotsb \prec N_n \preceq L_0 \prec \dotsb \prec
      L_m \]
    of submodules of~$M$.  Any proper inclusion in the original chains
    must also be a proper inclusion in the concatenated chain.
    Thus,~$\ell(M) \geq n + m = \ell(N) + \ell(M/N)$.
  \item Let us consider~$N_0 = \braces{(0, n) \mid n \in N}$.  Then
    $(M \oplus N) / N_0 \cong M$~and~$N_0 \cong N$, which yields the
    claim by Statement~(i).
  \item Let~$M$ be a module generated by $n$~elements.  Hence $M$~is
    a linear image of the free module~$S^n$, which by iteration of
    Statement~(ii) satisfies~$\ell(S^n) = n \cdot \ell(S)$.  Thus, the
    claim follows directly from Lemma~\ref{finitemap}.  \qedhere
  \end{enumerate}
\end{proof}

Hence every finite-length ring has the property that all its finitely
generated modules also have finite length.  Naturally, there are other
semirings that enjoy this property.  Trivially, every finitely
generated semimodule over a finite semiring (such as the
\textsc{Boolean} semifield~$\BB$) is also finite and therefore of finite length.  The following definition establishes the property
just discussed, which is fulfilled in all rings and all finite
semirings.

\begin{defn}
  We say that $S$~is \textsc{Artinian} if every finitely generated
  semimodule has finite length.
\end{defn}

As demonstrated in the proof of Theorem~\ref{artinrings}(iii), in
order to establish that $S$~is \textsc{Artinian} it suffices to show
that free semimodules of finite rank have finite length.  Our naming
\textsc{Artinian} is a slight abuse of traditional notions since the
term is usually used to characterize those modules that satisfy
the Descending Chain Condition (i.e., every descending chain of
submodules terminates).  However, in rings these two notions
coincide.  Any ring that satisfies the Descending Chain
Condition~(DCC) also satisfies the Ascending Chain
Condition~(ACC)~\cite[p.~90, Theorem~8.5]{ATIYAH}, and any module that
satisfies both DCC~and~ACC has finite length~\cite[p.~77, Propositions
6.7~and~6.8]{ATIYAH}.  By Theorem~\ref{artinrings}(iii), all finitely
generated modules over a finite-length ring also have finite length.
The converse implication is trivial.  In general, for semirings this
equivalence need not hold (see footnote to
Lemma~\ref{counterexample}), but since the DCC is nowhere as important
for semirings as it is for rings, the authors believe that our use of terminology is harmless.

\textsc{Artinian} semirings retain a very convenient property of
endomorphisms of vector spaces, which will be crucial for our
approach.

\begin{thm}
  \label{inj=surj}
  Suppose that $S$~is {\upshape \textsc{Artinian}}, and let $M$~be a
  finite-length semimodule and~$\alpha \in \End(M)$.  Then $\alpha$~is
  surjective if and only if $\alpha$~is injective.
\end{thm}

\begin{proof}
  The proof simply combines the well-known facts that surjective
  endomorphisms of \textsc{Noetherian} modules are injective, and
  injective endomorphisms of modules that satisfy the Descending Chain
  Condition (i.e., \textsc{Artinian} in the traditional sense) are
  surjective.  These two facts are established here for our
  semimodules.

  We start with necessity. Suppose that $\alpha$~is surjective.  For
  every endomorphism~$\varphi \in \End(M)$, we let
  \[ \Ker \varphi = \braces{(u, v) \in M \oplus M \mid \varphi(u) =
      \varphi(v)}. \]
  Then $\Ker \varphi$~is a subsemimodule of~$M \oplus M$ by the
  linearity of~$\varphi$.  Let~$r = \ell(M \oplus M)$ and consider the
  chain
  \[ \0 \prec \Ker \alpha^0 \preceq \Ker \alpha^1 \preceq \dotsb
    \preceq \Ker \alpha^r. \]
  The first strictness is justified by~$\Ker \alpha^0 =
  \Ker \id_M = \braces{(u, u) \mid u \in M} \succ \0$.  Thus, by the
  finite length~$r$, there exists some~$0 \leq i < r$ such that~$\Ker
  \alpha^i = \Ker \alpha^{i+1}$.

  To prove injectivity, let~$u, v \in M$ such that~$\alpha(u) =
  \alpha(v)$.  Recall that compositions of surjective functions are
  surjective.  By the surjectivity of $\alpha$~and~$\alpha^i$, there
  exist~$x, y \in M$ such that $\alpha^i(x) = u$~and~$\alpha^i(y) =
  v$.  Consequently, $\alpha^{i+1}(x) = \alpha^{i+1}(y)$ and thus~$(x,
  y) \in \Ker \alpha^{i+1} = \Ker \alpha^i$ by our choice of~$i$.
  However, $(x, y) \in \Ker \alpha^i$ directly yields~$u = \alpha^i(x)
  = \alpha^i(y) = v$.  Hence, $\alpha$~is injective.

  We continue with sufficiency, so let $\alpha$~be injective.  We show
  for all~$j \in \N$ that the condition~$u \notin \im \alpha^j$
  implies~$\alpha(u) \notin \im \alpha^{j+1}$.  For the sake of a
  contradiction, suppose that~$j \in \N$ and~$u \in M \setminus \im
  \alpha^j$ are such that~$\alpha(u) \in \im \alpha^{j+1}$.  Clearly,
  there exists~$v \in M$ such that $\alpha(u) = \alpha^{j+1}(v) =
  (\alpha\alpha^j)(v) = \alpha \bigl(\alpha^j(v) \bigr)$.
  Next we utilize the injectivity of~$\alpha$ to conclude~$u =
  \alpha^j(v)$, which yields~$u \in \im \alpha^j$ and our desired
  contradiction.  Thus, $\alpha(u) \notin \im \alpha^{j+1}$.

 Suppose that $\alpha$~is
  not surjective.  Then there exists~$u \in M$ such that~$u \notin \im
  \alpha$.  A straightforward induction utilizing the statement proved
  in the previous paragraph can now be used to show that~$\alpha^j(u)
  \notin \im \alpha^{j+1}$ for all~$j \in \N$.  However, this yields
  that the chain  
  \[ M = \im \alpha^0 \succeq \im \alpha^1 \succeq \dotsb \succeq \im
    \alpha^j \succeq \dotsb \] 
  has infinitely many proper inclusions, which contradicts that
  $M$~has finite length.  Therefore, $\alpha$~must be surjective.  We
  note that for sufficiency we only used that $M$~has finite length
  (not that $S$~is actually \textsc{Artinian}).
\end{proof}

\section{Pseudoregular Endomorphisms}
\label{sec:pseudoregular}
At this point we have established sufficient background for our main
notion, pseudoregular endomorphisms, that will be successfully
utilized in our pumping lemmata.  The special properties that define
them are established in the next lemma.

\begin{lem}[see~\protect{\cite[Proposition~1]{REUT80}}]
  \label{pseudocharacterization}
  Let $M$~be a semimodule and~$\alpha \in \End(M)$.  The following are
  equivalent:
  \begin{enumerate}
  \item $\im \alpha = \im \alpha^2$.
  \item There exist~$\gamma, \beta \in \End(M)$ such that $\alpha =
    \gamma\beta$~and~$\im \beta = \im (\beta\gamma\beta)$.
  \end{enumerate}
\end{lem}

\begin{proof}
  \mbox{ }
  \begin{itemize}
  \item We start with the implication~(i)~$\to$~(ii).  To this end, we
    select $\gamma = \mathord{\id_M}$~and~$\beta = \alpha$ and observe
    that
    \[ \alpha = \mathord{\id_M}\alpha = \gamma\beta \qquad \text{and}
      \qquad \im \beta = \im \alpha = \im \alpha^2 = \im (\alpha
      \mathord{\id_M} \alpha) = \im (\beta \gamma \beta). \]
  \item For the converse implication~(ii)~$\to$~(i), let~$\gamma,
    \beta \in \End(M)$ such that~$\alpha = \gamma\beta$
    and~$\im \beta = \im (\beta\gamma\beta)$.  Then
    \[ \im \alpha^2 = \im(\gamma\beta\gamma\beta) = \gamma
      \bigl(\im(\beta \gamma \beta) \bigr) = \gamma (\im \beta) = \im
      (\gamma\beta) = \im \alpha. \qedhere \]
  \end{itemize}
\end{proof}

\begin{defn}
  Let $M$~be a semimodule.  An endomorphism~$\alpha \in \End(M)$
  satisfying the conditions of Lemma~\ref{pseudocharacterization} is
  called \emph{pseudoregular}.  
\end{defn}

\textsc{Reutenauer}~\cite[Proposition~1]{REUT80} provides further
characterizations of pseudoregular endomorphisms that hold for a
field~$S$.  It is worthwhile to consider the following consequence.
Let $\alpha$~be a nonzero pseudoregular endomorphism of a finite
dimensional vector space~$V$.  Then there exists~$k \leq \dim V$ and a
basis~$\mathcal{B}$ of~$V$ such that the matrix representation
of~$\alpha$ with respect to~$\mathcal{B}$ is a block matrix
\[ \begin{pmatrix}
    A
    & 0_{(n-k) \times k} \\
    0_{k \times (n-k)}
    & 0_{(n - k) \times (n - k)}
  \end{pmatrix},
\]
where $A$~is an invertible $k \times k$"~matrix and $0_{m \times
  n}$~is the $m \times n$"~zero matrix for every~$m, n \in \NN$. 

Using Theorem~\ref{inj=surj} we can adapt another characterization
mentioned in~\cite[Proposition~1]{REUT80} to \textsc{Artinian}
semirings.

\begin{lem}
  Suppose that $S$~is {\upshape \textsc{Artinian}}, and let $M$ be a
  semimodule that has finite length.  Then $\alpha \in \End(M)$ is
  pseudoregular if and only if~$\alpha_* \colon \im \alpha \to \im
  \alpha$, which is defined for every~$u \in \im \alpha$
  by~$\alpha_*(u) = \alpha(u)$, is an isomorphism.  If $S$~is a ring,
  then this is equivalent to~$\im \alpha \cap \ker \alpha = \0$.
\end{lem}

\begin{proof}
  Clearly, $\im \alpha = \im \alpha^2$~is equivalent to surjectivity
  of~$\alpha_*$, so the result follows from
  Theorem~\ref{inj=surj}.  If $S$~is a ring, then~$\im \alpha \cap
  \ker \alpha = \0$ is equivalent to injectivity of~$\alpha_*$, and
  thereby surjectivity.
\end{proof}

Next we show a generalization of~\cite[Theorem~2.2]{JAC80}.  The
general proof idea is largely unchanged, but the lack of vector space
structure requires some adjustments in the details.  The same theorem
can be shown for vector spaces in a much more straightforward manner
using linear recurrences (see~\cite[Lemma~1]{REUT80}), but as this
proof relies on the existence of characteristic polynomials of
endomorphisms, it cannot be directly adapted to more general
semirings.

\begin{thm}[{see~\protect{\cite[Theorem~2.2]{JAC80}}}]
  \label{sequence}
  Let $M$ be a semimodule such that its dual~$M^\vee$ has finite
  length.  Moreover, let~$\alpha \in \End(M)$ be pseudoregular, and
  let $f \in M^\vee = \Hom(M, S)$~and~$v \in M$.  We consider the
  sequence~$(s_k)_{k \in \N}$ of elements of~$S$ given for every~$k
  \in \N$ by
  \[ s_k = f \bigl(\alpha^k (v) \bigr). \]
  If~$s_1 \neq 0$, then~$s_k \neq 0$ for infinitely many~$k \in \N$.
  More precisely, at most~$\ell(M^\vee)$ values of~$s_k$ vanish in a
  row.
\end{thm}

\begin{proof}
  We prove this statement in three steps.
  \begin{enumerate}
  \item As before, we define~$\alpha_* \colon \im \alpha \to \im
    \alpha$ for every~$u \in \im \alpha$ by~$\alpha_*(u) = \alpha(u)$.
    Since $\alpha_*$~is surjective, we can find a right
    inverse~$\alpha^* \colon \im \alpha \to \im \alpha$ such
    that~$\alpha_*\alpha^* = \mathord{\id_{(\im
        \alpha)}}$.\footnote{In the most general setting, finding a
      right inverse of a surjective function requires the Axiom of
      Choice.  In all cases of interest to us, this is not necessary.
      If $S$~is \textsc{Artinian}, then $\alpha_*$~is bijective, so
      there is a unique both-sided linear inverse.  If $\im \alpha$~is
      free of finite rank, then it suffices to choose finitely many
      preimages for the free generators of~$\im \alpha$.}  Next, we
    define~$\rho$ to be the map that sends each element~$g \colon M
    \to S$ of~$M^\vee$ to its restriction~$g \restrict{\im \alpha}$
    to~$\im \alpha$; i.e., 
    \[ \rho \colon M^\vee \to (\im \alpha)^\vee \qquad \text{with}
      \qquad \rho(g) = g \restrict{\im \alpha} \] for
    all~$g \in M^\vee = \Hom(M, S)$.  Clearly, $\rho$~is linear, so
    $\im \rho$~has finite length by Lemma~\ref{finitemap}.  Fix
    some~$n_0 \in \NN$ and let~$f_i = \rho(f\alpha^{n_0 + i})$ for
    every~$i \in \N$.  Then
    \[ f_i = \rho(f\alpha^{n_0 + i}) = \rho(f\alpha^{n_0 +
        i})\mathord{\id_M} = \rho(f\alpha^{n_0 + i}) \alpha_*\alpha^*
      = \rho(f\alpha^{n_0 + i + 1}) \alpha^* = f_{i+1}\alpha^* \]
    for every~$i \in \N$.
  \item Let~$r = \ell (\im \rho) + 1$ and~$M_i = \angles{\{\seq
      fr{r-i}\}}$ be the subsemimodule of~$\im \rho$ that is generated
    by~$\{\seq fr{r-i}\}$ for every~$0 \leq i \leq r$.  We consider the
    chain~$M_0 \preceq M_1 \preceq \ldots \preceq M_r$.  Since~$r >
    \ell(\im \rho)$, at least one of these inclusions is not proper.
    Let~$0 < i \leq r$.  If~$M_{i-1} = M_i$, then~$M_i = M_{i+1}$, which we
    prove as follows.  Since~$M_i = M_{i-1}$,  there exist
    coefficients~$\seq \lambda0r \in S$ such that
    \[ f_{r-i} = \sum_{j = 0}^{i-1} \lambda_j f_{r-j} \]
    and thus
    \[ f_{r-(i+1)} = f_{r-i-1} = f_{r-i}\alpha^* = \Bigl(\sum^{i-1}_{j
        = 0} \lambda_j f_{r-j} \Bigr) \alpha^* = \sum_{j=0}^{i-1}
      \lambda_j f_{r-j-1} = \sum_{j=1}^i 
      \lambda_{j-1} f_{r-j}. \]
    Therefore, $f_{r-(i+1)} \in \angles{\{\seq f{r-1}{r-i}\}} \preceq
    M_i$, so we have~$M_{i+1} = M_i$ by the construction of~$M_i$.  A
    straightforward induction then proves that~$M_r = M_{r-1}$.  Hence,
    there are coefficients~$\seq \mu1r \in S$ such that
    \begin{equation}
      \label{eq:3}
      f_0 = \sum_{j = 1}^{r} \mu_j f_j. \tag{$\dagger$}
    \end{equation} 
  \item Finally, let~$s_1 \neq 0$.  Assume by way of contradiction
    that there are only finitely many~$k \in \N$ such that~$s_k \neq
    0$.  Then there is some~$n \in \N$ such that~$s_n \neq 0$ and~$s_k
    = 0$ for all~$k > n$.  In particular,~$s_{n + 1} = \dotsb =
    s_{n+r} = 0$.  Set~$n_0 = n-1$ and define~$f_i$ as above.  Then
    \[ s_n = f_0 \bigl(\alpha(v) \bigr) = \Bigl(\sum_{j = 1}^r \mu_j
      f_j \Bigr) \bigl(\alpha(v) \bigr) = \sum_{j = 1}^r \mu_j f_j
      \bigl(\alpha(v) \bigr) = \sum_{j = 1}^r \mu_j s_{n+j} = 0\]
    by~\eqref{eq:3}, which contradicts the choice of~$n$.  Therefore,
    there must be infinitely many~$k \in \N$ such that~$s_k \neq 0$.
    In particular, we have shown that at most~$r - 1 = \ell (\im \rho)
    \leq \ell (M^\vee)$ values of~$s_k$ can vanish in a row. \qedhere
  \end{enumerate}
\end{proof}

We note that the previous proof relies crucially on the commutativity
of~$S$, since $M^\vee$~need not be a semimodule in the non-commutative
case.  Semimodules of finite length allow us to determine that
an endomorphism is pseudoregular simply by looking at its
factorizations.  We will later use a statement of this kind for the
proof of our pumping lemma.  However, one similar proposition can already be
adapted directly from the theory of vector spaces without any further work.

\begin{lem}[see \protect{\cite[Proposition 2.1]{JAC80}}]
  \label{pseudopower}
  Let $M$~be a finite-length semimodule and $\alpha \in \End(M)$.
  Then $\alpha^{\ell(M)}$~is pseudoregular. 
\end{lem}

\begin{proof}
  Consider the chain
  \[ M = \im \alpha^0 \succeq \im \alpha^1 \succeq \im \alpha^2 \succeq
    \ldots \succeq \im \alpha^{\ell(M)} \succeq 0. \] 
  By definition at least one of these inclusions is not
  proper.  Let~$0 < i \leq \ell(M)$.  If~$\im \alpha^{i-1} = \im
  \alpha^{i}$, then indeed also~$\im \alpha^{i} = \im \alpha^{i+1}$,
  so by another straightforward induction we also obtain~$\im
  \alpha^{\ell(M)} = \im \alpha^{2\ell(M)}$, which yields that
  $\alpha^{\ell(M)}$~is pseudoregular.  If only the last inclusion is
  improper (i.e., $\im \alpha^{\ell(M)} = \0$), then
  $\alpha^{\ell(M)}$~is the zero morphism and thereby trivially
  pseudoregular as well.  This concludes all cases and in each case
  $\alpha^{\ell(M)}$~is pseudoregular.
\end{proof}

\section{Pumping Lemmata}
\label{sec:pumping}
In this final section, we combine our derived results to
provide a pumping lemma for recognizable weighted languages.  In
general, pumping lemmata are used to prove that a (weighted) language
is not recognizable.  For illustration, we recall the classical
pumping lemma for recognizable languages, which is the main tool to
prove that a given language is not recognizable~\cite{yu97}.

\begin{thm}[see~\protect{\cite[Lemma~2]{PUMP}}]
  \label{thm:classic}
  Let $L$~be a recognizable language.  Then there exists~$n \in \N$
  such that for every~$w \in L$ with~$\abs w \geq n$ there is a
  factorization~$w = uxv$ with~$x \neq \varepsilon$ such that $ux^k v
  \in L$ for all~$k \in \N$. 
\end{thm}

Next, we show a similar result for recognizable weighted languages,
which was originally proven for fields in~\cite[Theorem~5]{JAC80},
although we adapted our proof using the ideas
of~\cite[Theorem~2]{REUT80}.  These ideas directly yield the
basic approach using our Theorem~\ref{sequence}.  Given a linear
representation~$(Q, \mathord{\inn}, \mathord{\out}, 
\mu)$ of a weighted language~$L \colon \Sigma^* \to S$ such that
(i)~$S^Q$~has finite length, (ii)~$\mu(x)$~is pseudoregular for
some~$x \in \Sigma^*$, and (iii)~$uxv \in \supp L$ for some~$u, v \in
\Sigma^*$, then for infinitely many~$k \in \N$,
\[ L(ux^kv) = \inn \cdot \mu(ux^k v) \cdot \out \neq 0. \]
We use that~$S^Q \cong (S^Q)^\vee$; i.e., that
$S^Q$~and~$(S^Q)^\vee$ are isomorphic, yielding finite length for $(S^Q)^\vee$.  Additionally, we note that we
do not conclude that~$L(ux^kv) \neq 0$ for all~$k \in \N$ (as in
Theorem~\ref{thm:classic}), but rather the inequality only holds for
infinitely many~$k \in \N$.  However, to make this approach 
applicable to any recognizable weighted language, we still need to
identify suitable conditions that enforce that a given
word~$w \in \Sigma^*$ contains a nontrivial subword~$x \in \Sigma^*$
with pseudoregular image~$\mu(x)$.  A simple combinatorial argument
following~\cite{REUT80} shows that if $w$~is long enough, then
there always exists a factorization~$w = uxy$ with~$x \neq \varepsilon$
such that $\mu(x)$~is pseudoregular.

\begin{defn}[see~\protect{\cite{REUT80}}]
  Let $\Sigma$~be a finite alphabet, $w \in \Sigma^*$, and~$n \in \N$.
  We recursively define when $w$~is a \emph{quasipower of order~$n$}.
  \begin{enumerate}
  \item If $n =
    0$~and~$w \neq \varepsilon$, then $w$ is a quasipower of order $0$.
  \item If $n > 0$~and~$w = uvu$ for
    some~$u, v \in \Sigma^*$ such that $u$~is a quasipower of
    order~$n-1$, then $w$ is a quasipower of order $n$. \qedhere
  \end{enumerate}
\end{defn}

Next, we recall that given any order~$r \in \N$ we can identify a
bound~$N_r$ such that words whose length is at least~$N_r$ necessarily
contain a quasipower of order~$r$.  Indeed the constant~$N_r$ can be
recursively defined for every~$r \in \N$ by
\[ N_0 = 1 \qquad \text{and} \qquad N_{r+1} = N_r \cdot (1 +
  \abs{\Sigma}^{N_r}). \]  

\begin{lem}[see~\protect{\cite[IV.~5]{SCHUTZ61} as cited in
    \cite[Lemma~2]{REUT80}}]
  \label{quasipower} 
  Let $\Sigma$~be a finite alphabet and~$r \in \N$.  There exists an
  integer~$N_r \in \N$ such that every word~$w \in \Sigma^*$ with
  $\abs{w} \geq N_r$ contains a subword that is a quasipower of
  order~$r$.
\end{lem}

Next, still following~\cite{REUT80}, we show that quasipowers of
suitably large order are sufficient to establish the existence of a
subword~$x$ such that~$\mu(x)$ is pseudoregular.

\begin{lem}[see~\protect{\cite{JAC80} as cited
    in~\cite[Theorem~1]{REUT80}}]
  \label{quasi=pseudo}
  Let $\Sigma$~be a finite alphabet, $M$~a semimodule that has finite
  length, and~$\mu \colon \Sigma^* \to \End(M)$ a monoid homomorphism.
  Every word~$w \in \Sigma^*$ that is a quasipower of order~$r =
  \ell(M) + 1$ contains a subword~$x \neq \varepsilon$ such that
  $\mu(x)$~is pseudoregular.
\end{lem}

\begin{proof}
  Let $w \in \Sigma^*$~be a quasipower of order~$r$, and let~$u_r =
  w$.  There are words~$\seq u0{r-1}, \seq v1n \in
  \Sigma^{\scriptscriptstyle+}$ such that~$u_i = u_{i-1} v_i u_{i-1}$
  for all~$1 \leq i \leq r$.  Thus,
  \[ \im \mu(u_i) = \im \mu(u_{i-1} v_i u_{i-1}) \preceq \im
    \mu(u_{i-1}), \]
  so we obtain the chain
  \[ \im \mu(u_r) \preceq \im \mu(u_{r-1}) \preceq \ldots \preceq \im
    \mu(u_0) \] 
  of $r+1$~subsemimodules of~$M$.  Therefore, $\im \mu(u_i) = \im
  \mu(u_{i-1})$ for some~$1 \leq i \leq r$, which yields  
  \[ \im \mu(u_{i-1}) = \im \mu(u_i) = \im \mu(u_{i-1} v_i u_{i-1}) =
    \im \bigl( \mu(u_{i-1}) \mu(v_i) \mu(u_{i-1}) \bigr). \]
  By Lemma~\ref{pseudocharacterization}(ii) we obtain that $\mu(v_i)
  \mu(u_{i-1}) = \mu(v_iu_{i-1})$~is pseudoregular.  Hence,
  we set~$x = v_iu_{i-1}$ to complete the proof.
\end{proof}

Our pumping lemma now follows directly.  The next main theorem still
contains the technical restriction that the semimodule~$S^Q$ has
finite length, where $Q$~is the set of states of a linear
representation for a given recognizable weighted language.  A slightly
more direct statement is expressed in the corollary that follows the
next theorem.

\begin{thm}[see~\protect{\cite[Theorem~5]{JAC80} as cited
    in~\cite[Theorem~2]{REUT80}}]
  \label{pumping1}  
  Let $\Sigma$~be a finite alphabet.  Moreover, let $(Q,
  \mathord{\inn}, \mathord{\out}, \mu)$ be a linear representation for
  the weighted language~$L \colon \Sigma^* \to S$.  If $S^Q$~has
  finite length, then there exists an integer~$N \in \N$ such that for
  every~$w \in \supp L$ with~$\abs{w} \geq N$ there exists a
  factorization~$w = uxv$ with~$x \neq \varepsilon$ such that
  \[ \{ux^kv \mid k \in \N\} \cap \supp L \]
  is infinite.
\end{thm}

\begin{proof}
  Let~$r = \ell(M) + 1$ and~$N = N_r$ as in Lemma~\ref{quasipower}.
  Since~$\abs w \geq N$, the word~$w$ contains a quasipower of
  order~$r$ by Lemma~\ref{quasipower}, and by Lemma~\ref{quasi=pseudo}
  there exists a factorization~$w = uxv$ such that
  $x \neq \varepsilon$ and $\mu(x)$~is pseudoregular.  Moreover,
  $\mathord{\inn} \cdot \mu(u) \in
  (S^Q)^\vee$~and~$\mu(v) \cdot \mathord{\out} \in S^Q$.  By
  assumption we have
  \[ L_w = \mathord{\inn} \cdot \mu(u) \mu(x) \mu(v) \cdot
    \mathord{\out} \neq 0. \] 
  Since~$S^Q \cong (S^Q)^\vee$ and $S^Q$~has finite length, we can
  apply Theorem~\ref{sequence} to obtain that for infinitely many~$k
  \in \N$,
  \[ \bigl(\inn \cdot \mu(u) \bigr) \cdot \mu(x)^k \cdot
    \bigl(\mu(v) \cdot \out \bigr) \neq 0. \]
  Since~$\mu(x)^k = \mu(x^k)$, this completes the proof.
\end{proof}

By extending this theorem from its original statement for fields to
more general semirings, we have identified a unified framework for the
classical pumping lemma by \textsc{Rabin} and
\textsc{Scott}~\cite{PUMP} (for the \textsc{Boolean} semifield) and
the pumping lemma for recognizable weighted languages over fields by
\textsc{Jacob}~\cite{JAC80}.  In practice, it is useful to be able to
reason about recognizability without knowing the number of states a
potential linear representation might have, which makes the
requirement that $S^Q$~has finite length troublesome.  This can be
remedied by requiring our semiring~$S$ to be \textsc{Artinian},
which of course still subsumes all the cases covered by the already
mentioned pumping lemmata.

\begin{cor}[of~\protect{Theorem~\ref{pumping1}}]
  \label{cor:pump}  
  Let $\Sigma$~be a finite alphabet, $S$~be an {\upshape
    \textsc{Artinian}} semiring, and $L$ be a recognizable weighted
  language~$L \colon \Sigma^* \to S$.  Then there exists an integer~$N
  \in \N$ such that for every~$w \in \supp L$ with~$\abs{w} \geq N$
  there exists a factorization~$w = uxv$ with~$x \neq \varepsilon$ such that 
  \[ \{ux^kv \mid k \in \N\} \cap \supp L \]
  is infinite.
\end{cor}

\begin{exmp}
  Directly generalizing a classical example of a non-regular language,
  there is no recognizable weighted language~$L$ over an
  \textsc{Artinian} semiring such
  that~$\supp L = \braces{a^n b^n \mid n \in \N}$.  Suppose that there
  is an \textsc{Artinian} semiring~$S$ and a recognizable weighted
  language~$L \colon \{a,b\}^* \to S$ such
  that~$\supp L = \braces{a^n b^n \mid n \in \N}$.  By the Corollary
  of Theorem~\ref{pumping1} there exists~$N \in \N$ such that
  $w = a^Nb^N$ admits a decomposition~$w = uxv$
  with~$x \neq \varepsilon$ such
  that~$\{ux^kv \mid k \in \N\} \cap \supp L$ is infinite.  Obviously
  this is a contradiction since no suitable
  subword~$x \neq \varepsilon$ (consider the cases $x = a^m$,
  $x = b^m$, and~$x = a^m b^n$) exists.  We note that such a
  recognizable weighted language~$L$ over a
  \textit{non-commutative} semiring exists.
\end{exmp}

If we drop the assumption that the alphabet~$\Sigma$ is finite, then
we obtain a notion of recognizable weighted languages that is useful
when applying the same pumping techniques to weighted tree languages
(see, for example,~\cite[Theorem~9.2]{BERSTEL}).  One result that
would be an ideal candidate for extension to semimodules is recalled
next.  Its extension would immediately yield pumping lemmata of
various forms (e.g.~\cite[Theorem~4]{REUT80}).

\begin{thm}[see~\protect{\cite[Theorem 3]{REUT80}}] 
  Let $\Sigma$~be a (not necessarily finite) alphabet and $V$~a vector
  space of finite nonzero dimension.  There is an integer~$N$ such
  that for each homomorphism~$\mu \colon \Sigma^* \to \End(V)$, every
  word~$w \in \Sigma^*$ with~$\abs{w} \geq N$ contains a subword~$x
  \neq \varepsilon$ such that $\mu(x)$~is pseudoregular.
\end{thm}

Unfortunately, the proof of this theorem uses the relationship of
nonvanishing elements of exterior powers of~$V$ to their components'
linear independence.  This cannot be easily extended even to
(non-integral) rings.  We conclude this section by stating two weak
pumping lemmata for recognizable weighted languages over infinite
alphabets.

\begin{thm}
  Let $\Sigma$~be a (possibly infinite) alphabet and
  $L \colon \Sigma^* \to S$ be a recognizable weighted language with
  linear representation~$(Q, \mathord{\inn}, \mathord{\out}, \mu)$
  such that $S^Q$~has finite length~$N = \ell(S^Q)$.  If there
  exists~$w \in \supp L$ with~$w = ab^N c$, then the set~$\{ab^k c
  \mid k \in \N\} \cap \supp L$ is infinite.
\end{thm}

\begin{proof}
  By Lemma~\ref{pseudopower}, $\mu(b^N) = \mu(b)^N$~is pseudoregular.
  Then the claim follows exactly as in Theorem~\ref{pumping1}.
\end{proof}

\begin{thm}
  Let the semiring~$S$ be finite, $\Sigma$~a (possibly
  infinite) alphabet, and $L \colon \Sigma^* \to S$ be a recognizable
  weighted language with linear
  representation~$(Q, \mathord{\inn}, \mathord{\out}, \mu)$.  There is
  an integer~$N$ such that for every~$w \in \supp L$
  with~$\abs w \geq N$ there exists a factorization~$w = uxv$
  with~$x \neq \varepsilon$ such that
  \[ \{ux^kv \mid k \in \N\} \cap \supp L \]
  is infinite.
\end{thm}

\begin{proof}
  Since $\End(S^Q)$~is finite, we can reduce to the case of
  finite alphabets.  To this end, we define the
  relation~$\mathord{\sim} = \Ker \mu$ on~$\Sigma$ (where $\Ker \mu$ is defined as in Theorem \ref{inj=surj}).  
  Clearly, $\sim$~is an equivalence relation.  From each of the finite
  number of equivalence classes~$[m]$ we choose a
  representative~$r_m$.  Now, we let~$\Gamma = \braces{r_m \mid m \in
    \Sigma}$ and extend the mapping~$m \mapsto r_m$ to the unique
  monoid homomorphism~$\psi \colon \Sigma^* \to \Gamma^*$.  By
  definition of~$\sim$, it is obvious that~$\mu(w) = \mu(\psi(w))$ for
  all~$w \in \Sigma^*$.

  Since~$\abs{\End(S^Q)} \leq \abs{S}^{\abs{Q}^2}$ (consider
  matrices), we have~$\abs{\Gamma} \leq \abs{S}^{\abs{Q}^2}$.  Let $N$
  be as in Theorem~\ref{pumping1}.  For every~$w \in \Sigma^*$
  with~$\abs{w} \geq N$. there exists a factorization~$w = uxv$
  with~$x \neq \varepsilon$ and infinite
  \[ \{\psi(ux^kv) \mid k \in \N\} \cap \supp L. \]
  By definition of~$\psi$, it is clear that this implies the
  infiniteness of the set
  \[ \{ux^kv \mid k \in \N\} \cap \supp L. \qedhere \] 
\end{proof}

\bibliographystyle{eptcs}
\bibliography{bib}

\end{document}